\documentclass[11pt,a4paper]{article}

\usepackage[top=3cm]{geometry}
\usepackage{graphicx,amsmath,amssymb,amsthm}

\usepackage{todonotes}
\usepackage{algorithm,algorithmic,caption} 
\usepackage[shortlabels]{enumitem}
\makeatletter

\newcommand*\IsExactlyOneToken[1]{%
	TT\fi
	\ifx\@empty#1\@empty
	\expandafter\@EmptyCase
	\else
	\expandafter\@IsOnlyOneToken
	\fi
	#1\@@@
}
\@ifdefinable\@IsOnlyOneToken{\def\@IsOnlyOneToken#1#2\@@@{%
		\ifx\@empty#2\@empty
}}
\@ifdefinable\@EmptyCase{\def\@EmptyCase#1\@@@{% #1 for robustness
		\iffalse % since 0 != 1
}}

\makeatother

\newcommand*{\parenthesis}[1]{%
	\if\IsExactlyOneToken{#1}%
	\,#1%
	\else
	\nonscript\!\left(#1\right)%
	\fi
}
\usepackage{xpunctuate}
\usepackage[all,british]{foreign} % Typesets abbreviations like e.g.
\redefnotforeign[ie]{i.e\xperiodafter} % i.e and e.g should not be italicized
\redefnotforeign[eg]{e.g\xperiodafter}

\renewcommand{\leq}{\leqslant}
\renewcommand{\geq}{\geqslant}

\usepackage{soul}

\newtheorem{theorem}{Theorem}[section]

\newtheorem{lemma}[theorem]{Lemma}
\newtheorem{proposition}[theorem]{Proposition}

\newtheorem{definition}[theorem]{Definition}
\newtheorem{observation}[theorem]{Observation}

\newtheorem*{claim*}{Claim}

\graphicspath{{./figures/}}

\newcommand{\ignore}[1]{}

\newcommand{\Block}[2]{C_{#1,#2}}
\newcommand{\block}{block\xspace}
\newcommand{\blocks}{blocks\xspace}
\newcommand{\Path}{R}
\newcommand{\turn}{turn\xspace}
\newcommand{\turns}{turns\xspace}

\newcommand{\oddd}{boundary\xspace}
\newcommand{\oddds}{boundaries\xspace}
\newcommand{\way}{way\xspace}
\newcommand{\ways}{ways\xspace}

\newcommand{\leaf}{leaf\xspace}
\newcommand{\leaves}{leaves\xspace}

\newcommand{\POWER}{power} % Text version \power
\DeclareMathOperator{\power}{power}

\newcommand{\degree}[2]{deg_{#1}(#2)}
\newcommand{\upperbound}{2\ell^2 + 5\ell + 7}

\title{Longest paths in $2$-edge-connected cubic graphs}
\author{Nikola K. Blanchard\thanks{Institut de Recherche en Informatique Fondamentale \& Lorraine Research Laboratory in Computer Science and its Applications. \mbox{nikola.k.blanchard@gmail.com}} \and
Eldar Fischer\thanks{Faculty of Computer Science, Israel Institute of Technology (Technion), Haifa, Israel. \mbox{eldar@cs.technion.ac.il}} \and Oded Lachish\thanks{Birkbeck, University of London, London, UK. \mbox{oded@dcs.bbk.ac.uk}} \and Felix Reidl\thanks{Birkbeck, University of London, London, UK. \mbox{f.reidl@dcs.bbk.ac.uk}}
 }
\begin{document}
\maketitle
\begin{abstract}
  \noindent 
  We prove almost tight bounds on the length of paths in $2$-edge-connected cubic
  graphs.  Concretely, we show that (i) every $2$-edge-connected cubic graph of
  size $n$ has a path of length
  $\Omega\left(\frac{\log^2{n}}{\log{\log{n}}}\right)$, and (ii) there exists a
  $2$-edge-connected cubic graph, such that every path in the graph has length
  $O(\log^2{n})$
\end{abstract}

\section{Introduction}\label{sec:intro}

Finding long cycles in graphs of low degree but high connectivity has been a
fruitful and interesting line of research. Relevant to our work is the result by
Bondy and Entringer~\cite{bondy1980twoconnected} which shows that every
2-edge-connected cubic graph on $n$ vertices has a cycle of length $\Omega(\log
n)$. Lang and Walther showed that this result is  essentially
tight~\cite{lang1968longest}.  However, to the best of our  knowledge, the
question about the longest \emph{path} has not been explicitly answered. One might be
tempted to conjecture that the answer should be similar; that the best lower
bound we can find is of the order of  $\log n$. Surprisingly, we show that
there always exists a path of the order of $\log^2 n$ (ignoring factors of
order $\log\log n$).  Using a simple construction we show that this is tight
up to lower-order factors, even for the case of planar cubic graphs.

For our lower bounds we decompose the graph along edge-cuts of size two and use
the resulting recursive structure to construct a long path.  The recurrence
takes advantage the important fact that a \emph{3-connected} cubic graph~$G$
always contains a cycle of length at least $|G|^c$, where $c = \log(1+\sqrt{5})
\approx 0.69$, which was first shown by Jackson~\cite{jackson1986longest},
proving a conjecture by Bondy and Simonovits~\cite{bondy1980threeconnected}, and
was later improved by Bilinski \etal~\cite{bilinski2011circumference} to $c
\approx 0.75$ and by  Liu \etal~\cite{liu2018circumference} to $c = 0.8$.
Together these result show that, up to lower-order factors, every cubic
$2$-connected graph contains a path of order $\log^2 n$ and we cannot expect
a longer path in general.

\section{Preliminaries}\label{sec:prelim}
For an integer $n$ we use the notation $[n]$ to denote the set
$\{1,2,\dots,n\}$, and for integers $i,j$ we use the notation $[i,j]$ to denote
the set $\{i,i+1,\dots,j\}$.

All graphs considered in this paper will be finite, undirected, and
loopless; but may contain parallel edges. For sets and set members we will use 
the  notation $X - v := X \setminus \{v\}$ and $X - Y := X \setminus Y$.
For vertex sets~$X \subseteq V(G)$ we will frequently use the notation
$\bar X := V(G)\setminus X$ for its complement.

For a graph~$G = (V,E)$ we use the shorthand $|G| = |V(G)|$ for the number of
vertices in the graph, and we write $uv \in G$ to denote that the edge
$\{u,v\}$ is contained in $E(G)$. A graph is \emph{cubic} if every vertex~$v
\in G$ has degree exactly three, and it is \emph{$k$-edge connected} if the
removal of any $k$ edges does not disconnect the graph. 

We define a \emph{cut} of a graph $G$ to be any bipartition $(X,\bar X)$ of
$V(G)$. The \emph{cutset} of a cut $(X,\bar X)$ is the set of edges with one
endpoint in~$X$ and the other in $\bar X$. The \emph{size} of a cut
is the size of its cutset.
We call the start and end-vertex of a path its \emph{endpoints} and all other
vertices on it the \emph{internal} vertices. For $u,v \in G$, an
\emph{$u$-$v$-path} in $G$ is a path with endpoints $u$ and $v$. A path
\emph{avoids} an edge (vertex) if the edge (vertex) is not a part of the path.

\section{Upper bound}%
In this section we define a family of graphs $H_\ell$ in
which every path has length $O(\log^2 |H_\ell|)$. Every graph in the family
consists of two isomorphic binary trees whose leaves are
identified, plus an edge between the roots of the trees and
some edges between their leaves, see Figure~\ref{fig:UB1}.
Formally, for every integer $\ell$, we define $H_\ell = (V_\ell,E_\ell)$ with 
vertex set
\[
  V_\ell = \big\{v_{h,j} \mid h\in [-\ell,\ell] \mbox{ and } j\in [2^{\ell-|h|}]\big\},
\]
and edge set $E_\ell$ which contains
\begin{enumerate}
	\item the edges $v_{h,j}v_{h-1,2j-1}$ and $v_{h,j}v_{h-1,2j}$, for every $h\in [\ell]$ and $j \in [2^{\ell-1}]$,
	\item the edges $v_{h,j}v_{h+1,2j-1}$ and $v_{h,j}v_{h+1,2j}$, for every $h\in [-\ell,-1]$ and $j \in [2^{\ell-1}]$,
	\item the edge $v_{\ell,1}v_{-\ell,1}$ between the roots of the trees, and 
	\item the edges $v_{0,2j-1}v_{0,2j}$, for every $j \in [2^{\ell-1}]$, 
        between leaves of the trees.
  % \item One of the trees - a balanced binary-tree $T^+$, of depth $\ell$, rooted at $v_{\ell,1}$, where 
  % \begin{itemize}
  % \item all the vertices $v_{0,j}$ such that $j\in [2^{\ell-|h|}]$ are the leaves,
  % \item every vertex $v_{h,j}$ such that $h \in \{1,2,\dots,\ell\}$ and $j\in [2^{\ell-|h|}]$ is an internal vertex whose children are $v_{h-1,2j-1}$ and $v_{h-1,2j}$, and 
  % \end{itemize}
  % \item The other tree - a balanced binary-tree $T^-$, of depth $\ell$, rooted at $v_{-\ell,1}$, that is isomorphic to $T^+$ by inverting the sign on the first index of the subscript of the vertices of $T^-$. 
\end{enumerate}

\begin{figure}
  \centering\includegraphics[width=.5\linewidth]{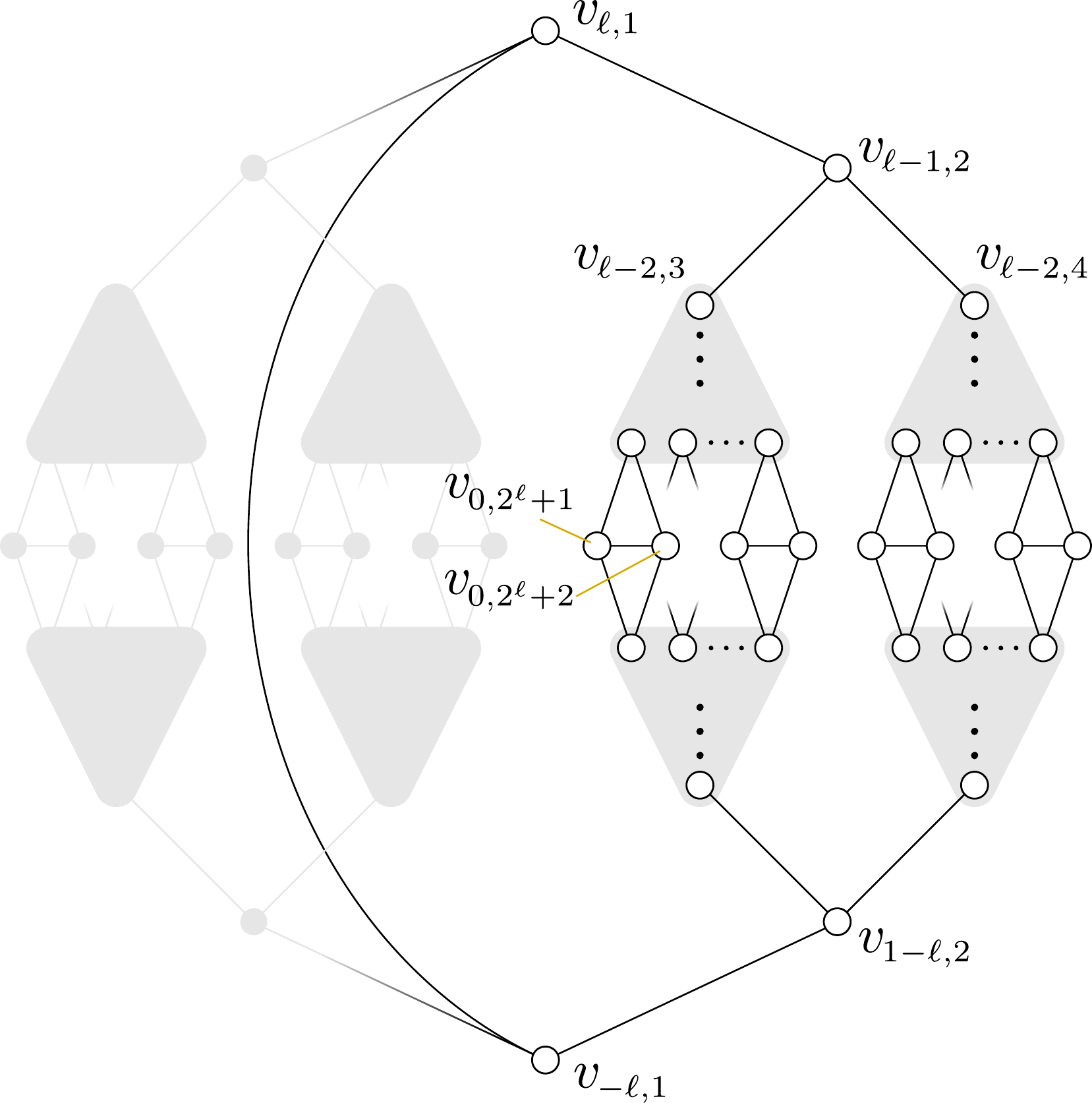}
  \caption{The graph $H_\ell$ with vertex exemplified vertex labels.}
  \label{fig:UB1}
\end{figure}

\noindent%
For every $h\in [\ell]$ and $j\in [2^{\ell-|h|}]$, we let 
$\Block{h}{j}$ be the induced subgraph of all the vertices in
the subtree of $T^+$ that is rooted in $v_{h,j}$ and  the subtree of $T^-$
that is rooted in $v_{-h,j}$. 
Formally, a $\Block{h}{j}$ is the induced subgraph of $G$ on the set of vertices $$\{v_{i,k}: i\in [-h,h], k\in [j2^{h-|i|},(j+1)2^{h-|i|}-1]\}.$$
\noindent
We call the subgraphs $\Block{h}{j}$
the \emph{blocks} of $H_\ell$.

\begin{lemma}\label{lem:G}
  $H_\ell$ is $3$-regular, planar, $2$-edge-connected and has $2^{\ell+1} + 2^{\ell} - 2$ vertices.
\end{lemma}
\begin{proof}
	By construction, $H_\ell$ is $3$-regular and, for every $h\in [-\ell,\ell]$, there are $2^{\ell-|h|}$ labelled vertices $v_{h,j}$. 
	Thus, the cardinality of $V$ is indeed $2^{\ell+1} + 2^{\ell} - 2$.

	We prove that $H_\ell$ is $2$-edge-connected by showing that for any pair of
  vertices $u$ and $v$ of $H_\ell$, there exists a cycle in $H_\ell$ that includes
	both $u$ and $v$.
  First assume that $u$ and $v$ lie inside the same tree, let's say $T^+$. 
  Let $w$ be the least common ancestor of $u$ and $v$ in $T^+$. Choose two
  paths $P_u$, $P_v$ starting both at $w$ and going through $u$ and $v$, 
  respectively, down towards some leaves of $T^+$. If we mirror $P_u$ and $P_v$
  into $T^-$ we obtain a cycle that contains both $u$ and $v$. 

  Now assume that $u$ and $v$ lie in different trees, say $u \in T^+$ and 
  $v \in T^-$. Let $v'$ be the mirror image of $v$ in $T^-$
  (so if $v = v_{-h,j}$ then $v' = v_{h,j}$). We construct the same cycle
  as in the previous case and simply note that it contains $v'$ and 
  therefore, by construction, also $v$.
  We remark that this also works when $v' = u$.

  A planar embedding is exemplified in Figure~\ref{fig:UB1} if we move
  the edge $v_{\ell,1}v_{-\ell,1}$ either into the middle of the figure
  or route it through the outer face.
\end{proof}

\noindent
To analyse the length of paths in $H_\ell$ we introduce the following terminology. Let $\Path$ be a path of $H_\ell$. 
We classify every vertex $v_{h,k}$ of $\Path$ as follows:
\begin{itemize}
  \item a \emph{\oddd} if $\degree{\Path}{v_{h,k}}$ is odd,
  \item a \emph{\leaf} if $h=0$ and $\degree{\Path}{v_{h,k}}$ is even
  \item a \emph{\turn} if $h\neq 0$, $\degree{\Path}{v_{h,k}}=2$ and in $\Path$ it is adjacent to both vertices $v_{h',2j-1}$, $v_{h',2j}$, where $|h'| = |h|-1$,
  \item a \emph{\way} if $h\neq 0$, $\degree{\Path}{v_{h,k}} =2$ and it is not a \turn.
\end{itemize}
We write just \oddd, \turn or \way, when the relevant path is clear from context.  
See Figure~\ref{fig:UB2} for a depiction of the vertex types.

\begin{figure}
  \centering\includegraphics[width=.6\linewidth]{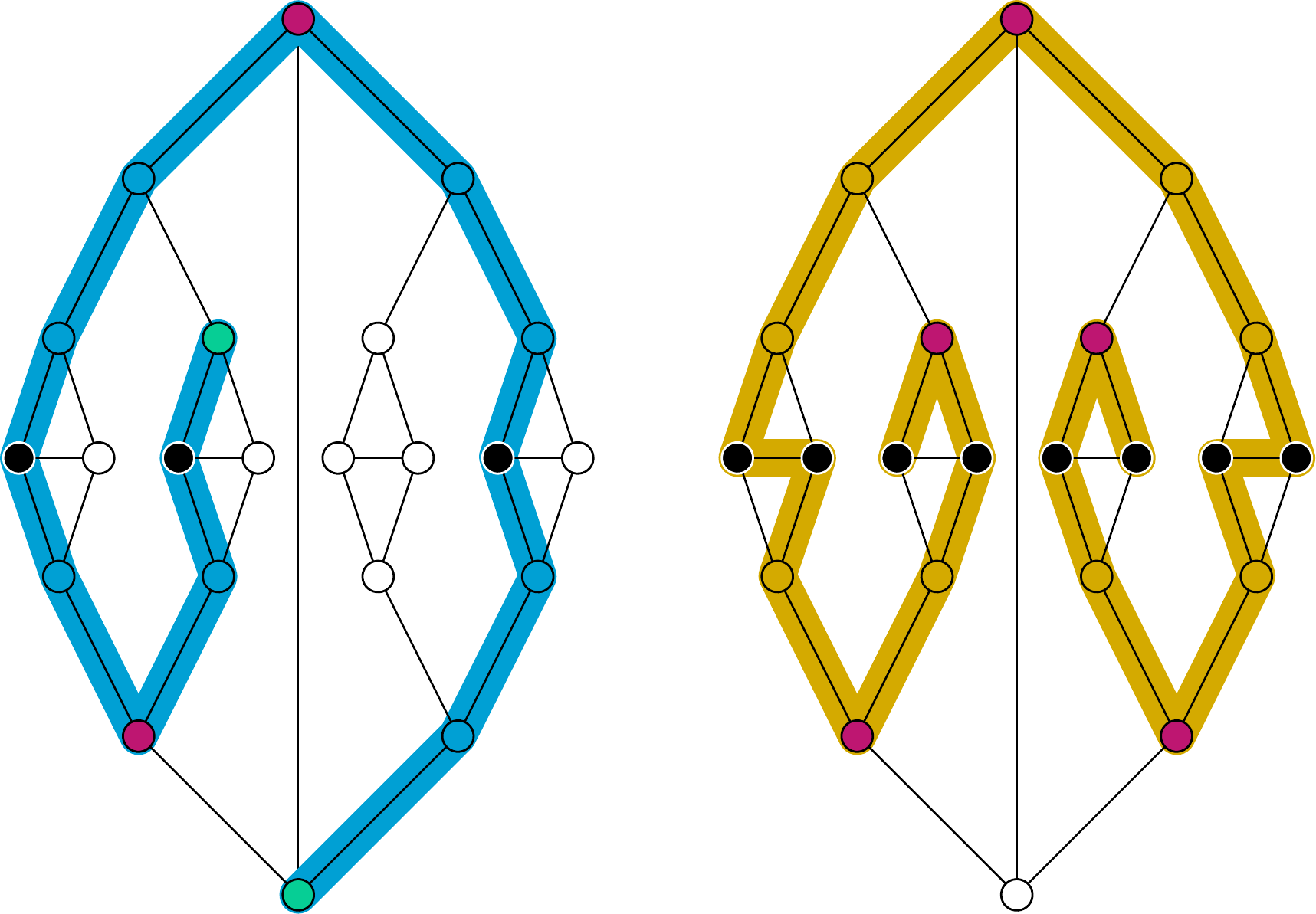}
  \caption{\label{fig:UB2}%
  A sub-optimal (left) and optimal (right) path in $H_3$ with vertices coloured
  by type: \leaves are black, \oddds green, \turns magenta, and \ways have the path's colour.}
\end{figure}

\begin{observation}\label{obs:boundaries}
  If $\Path$ is a path of $H_\ell$, then every one of its vertices is either a \leaf, a \oddd, a \turn, or a \way, and  
  the number of \oddd vertices is $0$ or $2$.
\end{observation}
\begin{proof}
  The observation follows from the fact that, as a standalone graph, $\Path$ is Eulerian.
\end{proof}

\noindent
For the remainder of this section let $\Path$ be an arbitrary path in $H_\ell$.
Recall that our goal is to upper bound the number of edges in $\Path$.
To do so we show that every edge in $\Path$ which is not between \leaf vertices 
can be charged to a specific \emph{segment} of $R$. A  segment is a minimum sub-path of $R$ with one endpoint being either at a \turn or a \oddd and the other either a \oddd or a \leaf. 
We further on bound the number of segments in $R$.
This is sufficient  for bounding the path's length because the number of edges between \leaves is bounded from above by the number of such segments.

We note that since there are at most two \oddd vertices each of degree $3$, the set of specific segments contains at most six paths with one of their ends being a \oddd vertex. The length of every one of these paths is bounded above by $O(\ell)$.
Thus, for our bound we need to focus on the paths between a \turn vertex and a \leaf. 
We note that every \turn vertex $v_{h,k}$, is an endpoint of at most two segments. 
Consequently, to prove this section's main result, we show that $\Path$ has at most two \turn vertices with the same first index.

\begin{lemma}\label{lem:turns}
Let $\Path$ be a path of $H$ and $h\in [-\ell,\ell]\setminus\{0\}$.
There exist at most $2$ \turn vertices $v_{h',j}$ of $\Path$ such that $|h'| = h$.
\end{lemma}
\begin{proof}
Assume for the sake of contradiction that there exist three distinct  \turn vertices $v_{h_1,j_1}$,$v_{h_2,j_2}$ and $v_{h_3,j_3}$ such that $|h_i|=h$ and $j_i\in [2^{\ell-h}]$, for every $i\in [3]$.

If any two of $j_1,j_2$ and $j_3$ are both equal to some $j$,  then two of
the three vertices are labelled $v_{h,j}$ and $v_{-h,j}$. Since $v_{h,j}$ and $v_{-h,j}$ are \turn vertices, all the edges of $\Path$ are contained in $\Block{h}{j}$,
because the only edges between vertices in $\Block{h}{j}$ and the vertices not in $\Block{h}{j}$ are adjacent to $v_{h,j}$ and $v_{-h,j}$ and are not in $R$, since $v_{h,j}$ and $v_{-h,j}$ are \turn vertices. 
Therefore, there can be no other \turn vertex whose first index has an absolute value of $h$,
contradicting our assumption on $v_{h_1,j_1}$,$v_{h_2,j_2}$ and $v_{h_3,j_3}$.

Now assume that each of the vertices $v_{h_1,j_1}$,$v_{h_2,j_2}$ and
$v_{h_3,j_3}$ are contained in a distinct \block, e.g. all three indices
$j_1,j_2,j_3$ are distinct. We next prove that this implies that every one of
these \blocks has a \oddd vertex, contradicting
Observation~\ref{obs:boundaries}.

Fix $i\in [3]$. We show that $\Block{h_i}{j_i}$ has a \oddd vertex. Assume
towards a contradiction it does not. Thus, in particular, neither one of $v_{h_i,j_i}$
and $v_{-h_i,j_i}$ is a \oddd vertex. Since $\Path$ is connected,
one of $v_{h_i,j_i}$ and $v_{-h_i,j_i}$ is a \way\ and the other is a \oddd
vertex or a way. Let $\Path'$ be subgraph consisting of the edges common to $\Path$
and $\Block{h_i}{j_i}$.  We observe that $v_{-h,j}$ is an vertex of $\Path'$ and
has degree $1$ in $\Path'$. Thus, by Observation~\ref{obs:boundaries},
$\Path'$ has another vertex of odd degree in $\Path'$. Since this vertex is
different from $v_{h_i,j_i}$ and $v_{-h_i,j_i}$, it also has odd degree in
$\Path$. Thus, $\Block{h_i}{j_i}$ has a \oddd vertex in contradiction to the
assumption that it does not. 

Hence our assumption on the existence of $v_{h_1,j_1}$,$v_{h_2,j_2}$ and
$v_{h_3,j_3}$ implies that $R$ has three boundary vertices, contradicting
Observation~\ref{obs:boundaries}. This proves the claim.
\end{proof}
\begin{theorem}\label{thm:upperBound}
	Every path in $H_\ell$ is of length at most $\upperbound$.  
\end{theorem}
\begin{proof}
  Fix some path $R$ in $H_\ell$, we show that it has length at most $\upperbound$.
	With every \turn $v_{h,j}$ we associate segments in $\Path$,
	both between $v_{h,j}$ and either a \oddd vertex or a \leaf vertex. One of
	the segments contains the vertex $v_{h',2j-1}$ and the other
	$v_{h',2j}$, where $|h'| = |h|-1$ and has the same sign as $h$. Every 
	internal vertex in either segment is a \way vertex, so the length of
	both paths is bounded above by $h$. Thus, by Lemma~\ref{lem:turns}, the
	total sum of segment-lengths associated with every \turn vertex on $\Path$
	is bounded above by $2\ell(\ell-1)$.

	Edges of $\Path$ may reside on a segment that includes both $v_{\ell,1}$ and
	$v_{-\ell,1}$. Such segment is either between two \oddd vertices, or a \oddd
	vertex and a \leaf vertex. The length of such a path is bounded above by
	$2\ell+1$.

	All the edges which are not part of the segments above are as follows: (i)
	adjacent to $2$ \leaf vertices, or (ii) in paths that do not include
	$v_{\ell,1}$ and $v_{-\ell,1}$ and are between a \oddd vertex and either a
	\oddd vertex or a \leaf vertex.

	There are at most $4\ell$ edges as in (ii). Every edge as described in (i) is
	adjacent to at most $2$ distinct segments, each one between a \oddd or \turn vertex and a \leaf vertex.
	By Observation~\ref{obs:boundaries} and Lemma~\ref{lem:turns}, there are at
	most $4\ell+6$ such paths. Thus, over all the number of edges as in (i) is
	bounded above by $4\ell+6$. Consequently, the number of edges in $\Path$ is
	bounded above by $2\ell(\ell-1) + 2\ell+1 + 4\ell + 6=\upperbound$.
\end{proof}

\newcommand{\lengthBiasOrder}{\log^2\log n}
\newcommand{\lengthBias}{36\lengthBiasOrder}
\newcommand{\pathLengthDenom}[1]{1+4\log\log \parenthesis{#1}}
\newcommand{\pathLengthWB}[1]{\frac{\log^2\parenthesis{#1}}{\pathLengthDenom{#1}}}
\newcommand{\pathLength}[1]{\pathLengthWB{#1} - \lengthBias}
\newcommand{\numOfVerticesInH}{\log^6{n}}
\newcommand{\heavyPE}{1-\frac{1}{\log{k}}}
\newcommand{\twoCycle}[1]{2\log #1 - 2\log\log #1 -8}

\newcommand{\U}{\mathbf U}

\section{Lower bound}
We prove here that every $2$-connected cubic graph $G$ has a path of length
$\pathLength{|V(G)|}$. For the simplicity of the presentation we have not 
optimized the constants involved. We will need the following two
known lemmas.

\begin{lemma}[Based on Bondy~\cite{bondy1980twoconnected}]\label{lem:bondy1980twoconnected}
	Every $2$-connected graph $H$ with maximum degree $3$ has a cycle of
	length at least $4\log|H| -4\log\log|H| - 20$.
\end{lemma}

\begin{lemma}[Jackson~\cite{jackson1986longest}]\label{lemma:jackson}
	Let~$G$ be a $3$-connected graph on $n$ vertices and let~$e_1,e_2 \in
	E(G)$. Then $e_1$ and~$e_2$ are contained in a cycle of~$G$ of length at
	least $n^{\log_2(1+\sqrt{5})-1}  + 1 = \Omega(n^{0.694})$.
\end{lemma}

\noindent
We will further need the following simple observation:
\begin{observation}\label{obs:simple}
  A $3$-connected cubic graph is either simple or it is the multi-graph on
  two vertices with three edges.
\end{observation}
\begin{proof}
  Assume there exists a graph~$G$ with parallel edges between $a,b$ on more
  than two vertices. Then $a$ and $b$ must each have one further edge
  that leaves the set $\{a,b\}$. But then $(\{a,b\}, V(G)-\{a,b\})$ is a cut
  of size two, contradiction.
\end{proof}

\noindent 
We will also need the following simple proposition regarding $2$-connected
cubic graphs.
\begin{proposition}\label{prop:2Cycle}
	Let $G$ be a $2$-connected cubic multi-graph and $e_1$, $e_2$  edges in
	the graph. There exists a cycle in $G$ containing both $e_1$ and $e_2$.
  Furthermore, if$e_2$ is contained in a cycle of length $\ell$ then 
  $e_1$ is contained in a cycle of length at least $\ell/2 + 1$.
\end{proposition}
\begin{proof}
  Since $G$ is $2$-connected, every pair of vertices incident to the same
  edge have two edge-disjoint paths between them. Therefore, $G$ has a cycle
  $C_1$ containing $e_1$ and a cycle $C_2$ containing $e_2$. Let $\ell := |C_2|$
  be the length of $C_2$. If $e_1$ is in $C_2$, then the
  proposition holds. So, assume that $e_1$ is not in $C_2$.

  Suppose that $C_1$ and $C_2$ have common edges. There exists a 
  shortest path in $C_1$ that contains $e_1$ and is between two vertices in
  $C_2$. Using this path and the rest of $C_2$, we get a cycle that has
  $e_1$ and $e_2$ as edges, or possibly just $e_1$, but has length at least
  $\ell/2+1$ ($e_1$ may be in parallel to an edge of $C_2$), which implies
  the theorem.

  We note that if $C_1$ and $C_2$ have a common vertex, then they also share
  at least one edge, because $H$ is cubic. Hence we only need to deal with the
  case that $C_1$ and $C_2$ are vertex-disjoint.
  Let $x$ be a vertex in $C_1$ and $y$ a vertex in $C_2$. Since $H$ is 
  $2$-connected, there exist two edge-disjoint paths between $x$ and $y$. These
  paths each have a minimal subpath that contains a vertex from $C_1$ and
  from $C_2$ but none of their edges. Using these paths and $C_1$ and $C_2$,
  we get a cycle that has $e_1$ and $e_2$ as edges, or possibly just $e_1$,
  but has length at least $\ell/2+1$ ($e_1$ may be in parallel to an edge of
  $C_2$), which proves the theorem.
\end{proof}

\noindent 
For the remainder of this section we fix $G$ to be  an arbitrary but
sufficiently large $2$-connected graph  on $n$ vertices.  We next give three
more definitions and one proof that are essential for this section's main result.
Afterwards we explain how these notions work together.

\begin{definition}[Tombolo, tombolo-cut]\label{def:semiIsland}
	A cut-set of size exactly $2$ is called a \emph{tombolo}. 
	A \emph{tombolo-cut} of a graph is a cut whose cut-set is a tombolo.
\end{definition}

\begin{proposition}\label{prop:edge_disjoint}
	For every tombolo-cut $(U, \bar U)$ of $G$ the edges of its tombolo are
	vertex disjoint.
\end{proposition}
\begin{proof}
	Let $(U, \bar U)$ be a tombolo-cut and $\{a_1b_1, a_2b_2\}$ its tombolo,
	where $a_1,a_2\in U$. Assume, for the sake of contradiction that $a_1 =
	a_2$. This implies that $(U - a_1,\bar U + a_1)$ is a cut of size one,
	contradicting that $G$ is $2$-connected. Hence, $a_1 \neq a_2$ and, by a
	symmetric argument, $b_1 \neq b_2$.
\end{proof}

\noindent
The following definition of a virtual subgraph is somewhat similar to the more general
idea of a \emph{torso} in a graph decomposition: we take a subgraph and encode the
external connectivity by adding (virtual) edges.

\begin{definition}[Virtual subgraph, virtual edges and real edges]
  A pair of vertices is called a \emph{port-pair} if they are in the same set of
  the tombolo-cut and incident to the cut's edges.
	A \emph{virtual subgraph} of a multi-graph $G$ is a multi-graph $H$ such 
	that $V(H) \subseteq V(G)$ and $E(H)$ is obtained by taking
	\begin{enumerate}
		\item\label{item:redges} all the edges of the sub-graph of $G$ induced on $V(H)$, referred to as the real edges of $H$, and
		\item\label{item:pedges} an edge $ab$ (possibly parallel), called a \emph{virtual edge}, for every $a$ and $b$ that are the vertices of a port-pair of  a tombolo whose edges are both not in $E(H)$.
	\end{enumerate} 
\end{definition}

\noindent
For the sake of clarity, we will call a tombolo-cut in a virtual subgraph 
a \emph{virtual tombolo-cut}. The cut itself might or might not actually use
virtual edges.

\begin{definition}[Peninsula]
  A virtual subgraph $H$ is called a \emph{peninsula}, if $(V(H),V(G)-V(H))$ is
  a tombolo-cut or if $H = G$.
\end{definition}

\noindent 
The proof of this section's main result is based on an induction over the
number of vertices in a peninsula.  The next proposition enables us to use
Lemma~\ref{lem:bondy1980twoconnected} on peninsulas.  After that we provide an
extra definition that enables us to formally explain how peninsulas connect with
virtual graphs and how they can be used for constructing paths.

%% we need this because the 2-connected cycle lemma is for regular graphs (to the best of my knowledge)
\begin{proposition}\label{prop:oneEntry}
  Let $H$ be a peninsula.  If $H \neq G$, then it has exactly one virtual edge
  and the rest of its edges are real edges.  The graph obtained from $H$ by
  replacing parallel edges with a single edge is $2$-connected and has maximum
  degree $3$.
\end{proposition}
\begin{proof}
  If $H = G$, then the proposition holds, by the assumption on $G$, so we assume
  that $H\neq G$.
    
  By construction, the only vertices that are not adjacent to the same vertices
  in $G$ and in $H$, are the ones that are adjacent to the tombolo-cut
  separating $H$ from the rest of the graph. We note that, by
  Proposition~\ref{prop:edge_disjoint}, that these are a pair of distinct vertices
  $a, b \in H$.
  Thus, by the definition of a virtual subgraph, the single virtual edge $ab$ 
  is added to create $H$ from $G[V(H)]$. Since $G$ is cubic, all vertices
  in $H$ have degree exactly three. The simple graph $H^*$ obtained from $H$
  by removing potential parallel edges between $a$ and $b$ has therefore
  only vertices of degree $2$ and $3$.

  It is left to show that this graph $H^*$ is indeed $2$-connected. If $H^* = H$
  then the statement clearly holds, so assume otherwise. Consider any pair
  of vertices $x,y \in H$, since $H$ is $2$-connected there exist two 
  edge-disjoint paths connecting them. Since $H$ is cubic, at most one of
  them can use an edge between $a$ and $b$, thus the same paths exist in $H^*$
  and we conclude that $H^*$ is $2$-connected.
\end{proof}

\begin{definition}[Internal port and external port]
  The single virtual edge of a peninsula, that is not the whole graph, is call
  the peninsula's \emph{internal port}. If $H = G$ an arbitrary edge is fixed to
  be its \emph{internal port}. The virtual edge created by a peninsula $H$
  in a virtual subgraph (other than $H$) is called its \emph{external port}.
\end{definition}

\begin{lemma}\label{lem:islandDetour}
	Let $e = ab$ be the external port of a peninsula $H$.
	There exists an $a$-$b$-path in $G$ whose internal vertices lie entirely in $H$
	and that has length at least $\twoCycle{H}$.
\end{lemma}
\begin{proof}
	Let $a$ and $b$ be the vertices adjacent to $H$'s internal port.
	By Proposition~\ref{prop:oneEntry}, the graph $\hat H$ that we get by replacing every parallel edge in $H$ with a single edge is $2$-connected and has maximum degree $3$. It also still has an edge $ab$.
	Thus, by Lemma~\ref{lem:bondy1980twoconnected}, $\hat H$ has cycle of  length at least $4\log|H| -4\log\log|H| - 20$.
	This implies, by Proposition~\ref{prop:2Cycle}, that $\hat H$ has a cycle of  length at least $2\log|H| -2\log\log|H| - 9$ that includes $ab$.
    By construction, the same holds for $H$.
    Using this cycle and the tombolo separating $H$ from the rest of the graph, we conclude the existence of the path required for the lemma to hold.
\end{proof}

Peninsulas by themselves are not sufficient for our goal. We need a virtual subgraph that will enable us to use Lemma~\ref{lemma:jackson}. We next give an algorithm that provides us with such a virtual graph.

\paragraph{The core of a peninsula $H$.}
The \emph{core} a peninsula $H$, with internal port $p$, is a virtual subgraph $C$ obtained using the following algorithm:
% \begin{algorithm}[h]
% 	\caption{Core finding algorithm}

\noindent\hspace*{.025\textwidth}\begin{minipage}{.95\textwidth}
\rule{\textwidth}{1.5pt}\vspace*{-4pt}%
\captionof{algorithm}{Core finding algorithm}
\rule{\textwidth}{.8pt}\vspace*{-10pt}%
\begin{algorithmic}[1]
	\STATE Set $C_1$ to be $H$ and $p = xy$ to be $H$'s internal port
	\STATE Set $i$ to be $1$
	\WHILE {$C_i$ has a virtual tombolo-cut $(U,V(C_i)-U)$ without $p$ in the cut-set}	
		\STATE Increase $i$ by $1$
    \IF{$\{x,y\} \subseteq U$}\label{alg:Ci}
      \STATE Set $C_i$ to be the virtual subgraph induced by $U$
    \ELSE\label{alg:Cii}
      \STATE Set $C_i$ to be the virtual subgraph induced by $V(H)-U$
    \ENDIF
	\ENDWHILE
	\STATE Set $C$ to be $C_i$
\end{algorithmic}\vspace*{-5pt}%
\rule{\textwidth}{1.5pt}%
\end{minipage}
\bigskip

\noindent  
Algorithm 1 always terminates since at each the new graph has strictly
less vertices and the initial number of vertices is finite. For simplicity, we
will still speak of \emph{the} core of a peninsula with the understanding that
any virtual subgraph constructed by the above algorithm will work\footnote{Cores
are unique but we do not need that fact and hence do not prove it here.}.

\begin{lemma}\label{lem:IncP}
  Let $H$ be a peninsula with internal port $p$ and core $C$. Then $p \in E(C)$.
\end{lemma}
\begin{proof}
  Let $p = a_1b_1$ and let $a_2$,$b_2$ be
  the vertices adjacent to $H$'s external port, hence $a_1,b_1\in V(H)$
  and  $a_2,b_2\not\in V(H)$. Furthermore, let $H = C_1,\ldots,C_\ell = C$ be the 
  sequences of virtual subgraphs constructed by Algorithm 1.
  
  Since $C_1$ is $H$ we also know that $a_1,b_1\in V(C_1)$ and  $a_2,b_2\not\in
  V(C_1)$.  When constructing $C_{i+1}$ from $C_i$, according to
  Lines~\ref{alg:Ci} and~\ref{alg:Cii} of Algorithm 1,  it is ensured that
  $a_1,b_1\in  V(C_{i+1})$ and $a_2,b_2\not\in V(C_{i+1})$. Consequently at the
  end $a_1,b_1\in V(C)$ and  $a_2,b_2\not\in V(C)$ and therefore---by the
  definition of virtual edges---it follows that $p\in E(C)$.
\end{proof}

\noindent
For a core $C$ and a peninsula $S$ whose external port~$e$ lies in $E(C)$ we
say that $C$ and $S$ are \emph{linked} or \emph{linked via $e$}.
The following lemma proves that all peninsulas linked to a core are disjoint 
and separated by the core from each other.

Now we are ready to explain how the induction we use for the main result works.
We have a peninsula $H$ and its core $C$. 
Our induction assumption is that for every peninsula $\hat H$ with internal port $\hat p$ and every vertex $x$ adjacent to $\hat p$, $\hat H$ has a long enough path that avoids $p$ and has $x$ as an end point.
The idea is to find a path with the required properties in $C$ and turn it into a path in $H$, by replacing external edges in this path, with paths in their peninsulas.
 The following lemmas provide the means to show that this can be done. 

\begin{lemma}\label{lem:cubicAndVirtual}
  Let $H$ be a peninsula with internal port $p$ and core $C$. Then $C$ is cubic,
  $2$-connected and for every pair of distinct peninsulas $S_1$ and $S_2$ linked to
  $C$ it holds that $S_1$ and $S_2$ do not have edges between them
  and that they do not share vertices with each other or with $C$.
\end{lemma}
\begin{proof}
  Let $(C_i)_{i \in [\ell]}$ with~$C_1 = H$ and $C_\ell = C$ be the sequence of
  virtual subgraphs constructed by Algorithm 1. We prove, by induction
  on~$i$, that the above statement holds for all~$C_i$ and therefore in particular
  for $C$. We further prove that $C_i$ together with all peninsulas linked to it
  forms a partition of $V(G)$.
  
  Let us begin with $C_1 = H$. By construction and by
  Proposition~\ref{prop:oneEntry}, $H$ is cubic, $2$-connected and only has $p$
  as a virtual edge. If~$H = G$ then the $V(H)$ trivially partitions $V(G)$.
  Otherwise, $H$ is a peninsula linked to another peninsula via $p$ and
  both together partition $V(G)$ via the tombolo-cut between them. The claimed
  properties therefore hold trivially.
  
  This concludes the induction base and we now assume that the lemma statement
  holds for~$C_{i-1}, i \geq 2 $ and prove it for $C_i$.  Let $(U,V(C_{i-1})-U)$
  be the cut with a cut-set of size $2$ that was used in order to construct $C_i$.
  Assume without loss of generality that $V(C_i) = U$ (otherwise rename the cut).
  Let $x_1, y_1 \in V(C_i)$ be the vertices that are adjacent to the edges of the
  cut-set of $(U,V(C_{i-1})-U)$.
  
  Let $\U$ contain $U$ and all vertices of peninsulas linked to $C_i$.
  In the following, we will refer to these peninsulas as \emph{$\U$-islands}.
  By induction, $\bar \U := V(G)-\hat U$ contains $V(C_{i-1})-U$ and the vertices of all
  peninsulas which are linked to $C_{i-1}$ but not to $C_i$. We  will refer to 
  these peninsulas as \emph{$\bar \U$-islands}. ($\star$) Note that, by construction,
  there can be edges from $\bar \U$-islands into $U$ but there cannot be 
  any edges from $\U$-islands into $(V(C_{i-1})-U)$.
  
  Assume first that $(\U, \bar \U)$ is a tombolo-cut in $G$. Then $x_1$ and
  $y_1$ are adjacent to the tombolo edges and lie inside $\U$.
  Accordingly, there exists a virtual edge $x_1y_1$ in $C_i$ which does
  not exist in $C_{i-1}$. Since $C_{i-1}$ is, by induction, cubic and 
  $2$-connected, it follows that $C_i$ is, too, due to this additional virtual 
  edge. Therefore it suffices to show that indeed $(\U, \bar\U)$ is a tombolo-cut in order to proof that $C_i$ cubic and $2$-connected.
  
  To that end, we show that the cut-set of $(\U, \bar\U)$ in $G$ 
  is at most as large as the cut-set of $(U,V(C_{i-1})-U)$ in $C_{i-1}$,
  which contains $2$ edges. Since $G$ is $2$-connected, this implies that 
  $(\U, \bar\U)$ is a tombolo-cut.
  % \hl{Note that according to our
  % construction of $\U$, if an edge in the cut-set of $(\U, \bar{\U})$ is adjacent to a vertex in $U$ and a vertex in $V(C_{i-1})$, then
  % according to the definition of a virtual subgraph this edge is also in
  %  $C_{i-1}$.}
  
  Consider any edge~$xy$ from the cut-set of $(\U, \bar\U)$ 
  in $G$ with $x \in \U$ and $y \in \bar\U$. By the induction
  assumption, $xy$ cannot connect two distinct peninsulas linked to
  $C_{i-1}$. Since every peninsula linked to $C_{i-1}$ is, by 
  construction, contained completely in either $\U$ or $\bar\U$,
  we further conclude that $xy$ cannot lie entirely within a single
  $\U$-island or $\bar\U$-island. As observed above ($\star$), $\U$-islands have no 
  edges into $(V(C_i)-U)$ and therefore $x$ cannot lie inside an $\U$-island,
  which implies that $x \in U$.
  
  If $y \in (V(C_{i-1}) - U)$, then $xy$ is an edge in $G$
  and therefore a real edge in $C_{i-1}$, hence it is also contained in the
  cut $(U, V(C_{i-1})-U)$ in $C_{i-1}$. This leaves the case in which
  $y$ is contained in a $\bar\U$-island $S$. Let $xy'$ be the external port
  of $S$ in $C_{i-1}$. Since $S$ is a $\bar\U$-island, $y'$ must lie in
  $V(C_{i-1}) - U$ and, by construction of $C_{i-1}$,
  $xy'$ is a virtual edge in $E(C_{i-1})$. We charge the virtual $xy'$ to
  the real edge $xy$.
  We conclude that every distinct edge in the cut-set of $(\U, \bar\U)$
  implies a distinct edge in the cut-set of $(U,V(C_{i-1})-U)$.

  Finally, let us proof that all peninsulas linked to $C_i$ are pairwise
  vertex-disjoint, not connected to each other and disjoint from $C_i$. We
  proved above that the cut $(\U,\bar\U)$ is actually a tombolo-cut.
  Accordingly, all $\bar\U$-islands together with $(V(C_{i-1}) - U)$ form a
  single large peninsula $S$ linked to $C_i$ via $x_1y_1$ while all
  $\U$-islands are linked to $C_i$ the same way they were linked to $C_{i-1}$.
  Thus the claim still holds among all $\U$-islands and we have to only consider
  cases involving the newly linked peninsula $S$. We already observed that
  $\U$-islands have no edges into $(V(C_{i-1}) - U)$ and, by induction, no edges
  into $\bar\U$-islands, thus they have no edge into $S$ and are disjoint from
  it. The peninsula $S$ is, by construction via a cut in $C_{i-1}$, disjoint
  from $U = V(C_i)$ and the claim that $U$ together with $S$ and all
  $\U$-islands partitions $V(G)$ follows from the same construction.
\end{proof}

\begin{lemma}\label{lem:Core2}
	Let $H$ be a peninsula with internal port $p$ and core $C$, where $|C| > 2$. 
	Then $C$ is either $3$-connected and does not have parallel edges, or 
  $C$ is $2$-connected with a two-cut $(U,V(C)-U)$ that contains $p$ in its cut-set
  and in which the virtual subgraphs induced by $U$ and $(V(C)-U)$ 
  are both cubic and $3$-connected.
\end{lemma}
\begin{proof}
  Since Algorithm 1 iterates as long as there is a virtual
  tombolo-cut that
  does not contain $p$ it is clear that it terminates if either $C_i$ is
  $3$-connected or if every  virtual tombolo-cut contains $p$. Let us first
  show that there is actually only a single such cut.
  
  Assume towards a contradiction that there exist multiple virtual tombolo-cuts
  with $p$ in them. Note that the respective other edge is a bridge in the graph
  $C-p$. Let $e, e'$ be bridges in $C-p$ which both have an endpoint in some
  $2$-connected component $X$ of $C-p$. Then $X$ is $(X,V(C)-X)$ is a tombolo-cut
  which does not contain $p$ in its tombolo---a contradiction. We
  conclude $(U, V(C)-U)$ is the only virtual tombolo-cut in $C$
  and we let $\{p=x_1y_1,x_2y_2\}$ be its tombolo.
  
  Since every other cut of $C$ must have size three or larger, its directly
  follows that the virtual subgraphs induced by $U$ (with the additional virtual
  edge $x_1x_2$) and  $(V(C)-U)$ (with the additional virtual edge 
  $y_1y_2$) are cubic and $3$-connected.
\end{proof}

\begin{lemma}\label{lem:CorePath}
  Let $H$ be a peninsula, $p$ its internal port and $C$ its core.  For every
  vertex $x$ adjacent to $p$ there exists a path in $C$ of length at least
  $|C|^{0.69}/2-1$ that avoids $p$ and has $x$ as one of its endpoints.
\end{lemma}
\begin{proof}
  By Lemma~\ref{lem:Core2}, the core $C$ is either $3$-connected or has a special
  cut. Suppose first that $C$ is $3$-connected. If $C$ is the multigraph with
  two vertices and three edges, then there is trivially a cycle of length at
  least $2|C|/3 = 4/3 > 2^{0.69}/2-1$ that includes $p$. Otherwise, by
  Observation~\ref{obs:simple}, $C$ is simple and by Lemma~\ref{lemma:jackson}
  has a cycle of length at least $|C|^{0.69}$ that includes $p$. In either case,
  the respective cycle shows that for every vertex $x$ adjacent to $p$ there
  exists a path as required in the lemma.  
			
  Now assume that $C$ is not $3$-connected but instead can be partitioned into
  two sets $U$ and $U'$ such that the virtual subgraphs  induced by these sets
  are both $3$-connected.  At least one of the sets has $|C|/2$ vertices,
  without loss of generality assume that this is $U$ (otherwise rename the sets
  accordingly).  Let $\hat C$ be the virtual subgraph whose vertices are the set
  $U$ and let $p'$ be the virtual edge in $\hat C$ that includes a vertex adjacent
  to $p$. If $\hat C$ is simple, then by Lemma~\ref{lemma:jackson} it has a
  cycle of length at least $|\hat C|^{0.69}$ that includes $p'$; if $C$ has
  only two vertices then trivially it  has a cycle of length at least $2|\hat
  C|/3$ that includes $p'$.  According to the construction of $p'$, this cycle
  implies the existence of a cycle in $C$ that includes $p$ and has length at
  least $(|\hat C|/2)^{0.69}$, which in turn implies the existence of a path as
  stated by the lemma.
\end{proof}

\begin{lemma}\label{lemma:OneEdge}%% because of the not including $p$
	Let $H$ be a peninsula, $p$ its internal port, $C$ its core and $e \neq p$ an
  	edge in $E(C)$. For every vertex $x$ adjacent to $p$ there exists a path in $C$
	that starts in $x$, avoids $p$ and has $e$ as its last edge.
\end{lemma}
\begin{proof}
  Since $C$ is cubic and $2$-connected, by
  Proposition~\ref{prop:2Cycle},  $e$ and $p$ are in some cycle in $C$.
  Consequently, for every vertex $x$ adjacent to $p$ there is a path in this cycle
  that does not include $p$ and can be seen as starting in $x$ and having $e$ as
  its last edge.
\end{proof}

\begin{lemma}\label{lemma:TwoEdges}
  Let $H$ be a peninsula, $p$ its internal port, $C$ its core and  $e_1$, $e_2$
  distinct edges in $E(C)-p$. For every vertex $x$ adjacent to $p$,
  there exists a path in $C$ that avoids $p$, starts
  in $x$, contains $e_1$ and $e_2$, and ends in either $e_1$ or $e_2$.
\end{lemma}
\begin{proof}
  Since $C$ is cubic and $2$-connected and, by
  Proposition~\ref{prop:2Cycle},  $e_1$ and $e_2$ are in some cycle in $C$. If
  $p$ is also in $C$ then for every vertex $x$ adjacent to $p$ there exists a
  path in this cycle that avoids $p$ and starts in $x$, contains both $e_1$ and $e_2$
  and ends in either~$e_1$ or~$e_2$.

  Thus assume that $p$ is not on the cycle and fix one of the endpoints of $p$
  as $x$. Since $C$ is $2$-connected there exist two edge-disjoint paths from
  $x$ to some arbitrary vertex on the cycle and at least one of these two
  paths avoids $p$. Therefore there exists a path starting at $x$, avoiding
  $p$, containing both $e_1$ and $e_2$ and ending in either $e_1$ or $e_2$,
  as claimed.
\end{proof}

\begin{definition}[\POWER]
	Let $H$ be a peninsula, $p = xy$ its internal port.
	We denote by $\power_s(H)$ the length of the longest path in $H$ that starts
  in $s \in \{x,y\}$ and avoids $p$ and we write
	\[
		\power(H) := \min_{s \in \{x,y\}} \power_s(H).
	\]
\end{definition}

\noindent 
The following lemma now relates the existence of a long path in the \emph{core}
of a peninsula---which might include virtual edges---to the existance of a
long path in the peninsula itself. The latter path, by definition, only
consists of real edges. We express this fact by proving recurrent
inequalities for $\power(\cdot)$.

\begin{lemma}\label{lemma:PathC}
	Let $H$ be a peninsula, $p$ its internal port, $C$ its core. Let
	further $P$ be a path in $C$ that avoids $p$ but starts in a vertex $x$
	that is incident to $p$ and has length at least $\power(H)$.
	
  Assume $P$ contains two virtual edges $e_1, e_2$ which link the peninsulas
  $S_1,S_2$ to $C$, respectively. Assume $e_1$ is the edge closer to $x$ on $P$.
  Then
  \begin{align*}
    \power(H) &\geq \twoCycle{|S_1|} +\power(S_2).
  \intertext{%
    If $P$ contains only one virtual edge~$e_1$ which links $S_1$ to $C$, then
  }
    \power(H) &\geq 1+\power(S_1).
  \end{align*}
\end{lemma}
\begin{proof}
  Note that $\power(H)$ concerns paths of the peninsulas $H$, hence in order
  to use $P$ to find lower bounds on $\power(H)$ we first need to argue that
  we can replace all virtual edges on $P$ by paths through peninsulas 
  that contain only real edges. 
	Let $\mathcal S$ be the collection of all peninsulas linked to $C$. By
  Lemma~\ref{lem:cubicAndVirtual}, all members of $\mathcal S$ are pairwise disjoint,
  not connected by edges and also disjoint from $C$. Therefore we can apply
  Lemma~\ref{lem:islandDetour} to each virtual edge $e \in P$ which 
  links $S \in \mathcal S$ to $C$ and replace it a real path of length
  at least $\twoCycle{|S|}$.
  
  It is left to show that we can make the above claimed guarantees on 
  $\power(H)$. In the first case, we replace $e_1$ as described above
  by a path of length $\twoCycle{|S_1|}$, however, we replace $e_2$ by 
  crossing the tombolo and finding a path of length $\power(S_2)$ inside
  of $S_2$, without re-surfacing through the tombolo again (hence all
  edges of $P$ after $e_2$ are lost). Note that such a path exists by the
  definition of $\power(\cdot)$, proving the first inequality.
  
  Similarly, if $P$ contains only one virtual edge $e_1$, we replace it
  by a path of length $\power(S_1)$ without resurfacing into $C$. The constructed
  path contains at least one edge more than $\power(S_1)$ since we count the
  tombolo-edge that leads into $S_1$. This proves the second inequality.
\end{proof}

\begin{theorem}\label{thm:LB}
  Let $G$ be $2$-connected and cubic. Then $G$ has a path of length at least
  $\pathLength{n} = O(\frac{\log^2 n}{\log\log n})$.
\end{theorem}
\begin{proof}
	We prove that for every peninsula $H$ on $k$ vertices $\power(H)\geq
	\delta(|H|)$, e.g. $H$ has a path of length at 
  least $\delta(k) := \pathLength{k}$. 
  Since $G$ is also a peninsula this implies the theorem.

	Let set $k = |H|$. Suppose first that $k \leq \numOfVerticesInH$. Then, we
	only need to show that $\power(H) \geq \delta(\numOfVerticesInH)$.
	The function $\delta(\numOfVerticesInH)$ is negative for $n \geq 2$,
	thus the inequality trivially holds. Therefore assume that
	$k > \numOfVerticesInH$ and hence we may assume that $k$ is 
	sufficiently large for the asymptotic inequalities we use in the sequel.
	By induction, assume that $\power(\hat H)\geq \delta(|\hat H|)$ for every peninsula $\hat H$ with  $|\hat H| < k$. We next prove the induction step.

	Let $p$ be $H$'s internal port, $x$ a vertex adjacent to $p$ and let
	$C$ be the core of $H$. By Lemma~\ref{lem:Core2}, $C$ is either
	$3$-connected or $V(C)$ can be partitioned into two set $U,U'$ such that
	the virtual subgraphs induced by $U$ and $U'$ are both $3$-connected. If
	$|C|> \log^{3}{k}$, then by Lemma~\ref{lem:CorePath} there exists a path
	in $C$ which starts in $x$, avoids $p$ and has length at least
	$|C/2|^{0.69}-1$. Consequently, by Lemma~\ref{lemma:PathC}, there exists a
	path in $H$ that includes only real edges, avoids $p$, has $x$ as an
	endpoint and has length at least $|C/2|^{0.69}-1> \log^{2}{k}/2 >
	\delta(k)$ and the theorem statement holds. Thus assume that $|C|<
	\log^{3}{k}$. 

	Since $\log^{3}{k} < k$, $E(C)$ has at least one virtual
	edge that is not $p$. Suppose that there exists a peninsula $S$
	linked to $C$ via $e \in E(C)$ such that $|S| > k(\heavyPE)$. Let $x$ be a vertex adjacent to $p$. By Lemma~\ref{lemma:OneEdge}, there exists path in $C$ 
	that starts in $x$, avoids $p$ and has $e$ as its last edge.
	Consequently, by Lemma~\ref{lemma:PathC}, 
	\[
		\power(H) \geq \delta\big(|S|\big)+1 
		 	   \geq \delta\Big(k\big( \heavyPE \big)\Big) + 1.
	\]
	We now show that the above implies $\power(H) > \delta(k)$:
  \begin{align*}
  		&\phantom{{}={}} \delta\Big(k\big(\heavyPE\big)\Big) + 1 + \lengthBias 
        = \pathLengthWB{k\left(\heavyPE\right)} + 1 \\
  		& > \frac{\Big(\log k - \frac{1}{\log k}\Big)^2}{\pathLengthDenom{k}}  + 1
  		  > \frac{\log^2k}{\pathLengthDenom{k}} 
            - \frac{2\log k (\log k)^{-1}}{\pathLengthDenom{k}} + 1 \\
  		& > \frac{\log^2k}{\pathLengthDenom{k}} = \delta(k) + \lengthBias.
  \end{align*}
	It remains to show that $\power(H) \geq \delta(k)$ in case where
	no such big peninsula exists.

	\newcommand{\SILB}{\frac{k}{2\log^4k}}
	Assume that every peninsula $S$ linked to $C$ via an edge other
	than $p$ contains less than $k(\heavyPE)$ vertices. Since $C$ is cubic,
	$|E(C)| = 1.5|C|$ which means that there are at most 
	$1.5|C| < 1.5 \log^{3}{n}$ peninsulas linked to $C$.

	Let $S_1, S_2$ be the two largest peninsulas attached to $C$ via $e_1, e_2
	\in E(C)-p$, respectively, with $|S_1| \geq |S_2|$. By averaging we conclude
	that $|S_1| \geq \frac{k-|C|}{1.5 \log^{3}{n}} > \frac{k}{2\log^3n}$.
	However, by our prior assumption, $|S_1| < k(\heavyPE)$. Therefore we can
	assert that 
	\[
		|S_2| \geq \frac{k - k(\heavyPE) -|C|}{1.5\log^{3}{k}} 
			  = \frac{\frac{k}{\log{k}}-|C|}{1.5\log^{3}{k}} > \SILB.
	\]
	Let $x$ be a vertex adjacent to $p$. By, Lemma~\ref{lemma:TwoEdges} there
	exists a path in $C$ which starts in $x$, avoids $p$ and contains
	both $e_1$ and $e_2$. Consequently, by
	Lemma~\ref{lemma:PathC},  
	\[
		\power(H)\geq \delta\Big(\SILB\Big) + \twoCycle{\SILB},
	\]
	where $\SILB$ is a lower bound on the sizes of both $S_1$ and $S_2$. The
	following computation implies that $\power(H)\geq\delta(k)$:
	\begin{align*}
		&\phantom{{}<{}} \delta\left(\SILB\right) + \twoCycle{\SILB} + \lengthBias \\
		&> \pathLengthWB{\SILB} + 2\log{k} - 2 - 8\log\log k - 2\log\log k  - 8\\
		&>  \frac{\left(\log k  - 1 - 4\log\log k\right)^2}{\pathLengthDenom{k}} + 2\log k - 10\log\log k - 10\\
		&> \frac{\log^2k}{\pathLengthDenom{k}} - \log k + 1 + 4\log\log k  + 2\log k - 10\log\log k - 10\\
		&> \frac{\log^2k}{\pathLengthDenom{k}} + \log k - 6\log\log k - 10
		 > \delta(k) + \lengthBias.
	\end{align*}
  This concludes the proof.
\end{proof}

\end{document}